\documentclass{article}
\usepackage{ijcai25}

\usepackage{times}
\usepackage{soul}
\usepackage{url}
\usepackage[hidelinks]{hyperref}
\usepackage[utf8]{inputenc}
\usepackage[small]{caption}
\usepackage{graphicx}
\usepackage{amsmath}
\usepackage{amsthm}
\usepackage{booktabs}
\usepackage{algorithm}
\usepackage{algorithmic}
\usepackage[switch]{lineno}

\usepackage{epstopdf}
\usepackage{multirow}
\usepackage{amsfonts}
\usepackage{amssymb}
\usepackage{soul}
\usepackage{color}
\usepackage{threeparttable}
\usepackage{subcaption}

\usepackage{epsfig}
\usepackage{array}
\usepackage{utfsym}
\usepackage{diagbox}

\newtheorem{proposition}{Proposition}

\title{Heterogeneous Value Decomposition Policy Fusion for Multi-Agent Cooperation}

\author{
Siying Wang$^1$
\and
Yang Zhou$^2$\and
Zhitong Zhao$^2$\and
Ruoning Zhang$^2$\and\\
Jinliang Shao$^1$\and
Wenyu Chen$^2$\and
Yuhua Cheng$^1$\thanks{Corresponding author}\\
\affiliations
$^1$School of Automation Engineering, University of Electronic Science and Technology of China\\
$^2$School of Computer Science and Engineering, \\
University of Electronic Science and Technology of China\\
\emails
siyingwang@uestc.edu.cn,
\{zhouy, zhaozhitong, zhangruoning\}@std.uestc.edu.cn,\\
\{jinliangshao, cwy, yhcheng\}@.uestc.edu.cn,
}

\begin{document}

\maketitle

\begin{abstract}
    Value decomposition (VD) has become one of the most prominent solutions in cooperative multi-agent reinforcement learning. Most existing methods generally explore how to factorize the joint value and minimize the discrepancies between agent observations and characteristics of environmental states. However, direct decomposition may result in limited representation or difficulty in optimization. Orthogonal to designing a new factorization scheme, in this paper, we propose Heterogeneous Policy Fusion (HPF) to integrate the strengths of various VD methods. We construct a composite policy set to select policies for interaction adaptively. Specifically, this adaptive mechanism allows agents’ trajectories to benefit from diverse policy transitions while incorporating the advantages of each factorization method. Additionally, HPF introduces a constraint between these heterogeneous policies to rectify the misleading update caused by the unexpected exploratory or suboptimal non-cooperation. Experimental results on cooperative tasks show HPF’s superior performance over multiple baselines, proving its effectiveness and ease of implementation.
\end{abstract}

\section{Introduction}\label{}

Multi-Agent Reinforcement Learning (MARL)~\cite{jaques2019social,mei2023mac}, derived from reinforcement learning methods~\cite{sutton_reinforcement_2018}, has proven to be successful in various tasks that involve cooperation and collaboration among multiple agents, such as coordination of robots~\cite{perrusquia2021multi}, autonomous Unmanned Aerial Vehicles (UAVs)~\cite{9712866}, and traffic signal control~\cite{yang2023causal}. These tasks often require agents to possess the ability to make distributed sequential decisions while considering collaboration with each other to collectively achieve a common goal. However, learning such effective cooperation policies among these distributed agents remains a major challenge due to the partial observability and non-stationarity in the environment resulting from the changing policies~\cite{oliehoek_concise_2016,tan1993multi}.

A fully centralized information processing agent seems to be a good solution to address the above issues. Nevertheless, it encounters the challenge of an exponentially expanding joint action space with the increase in the number of collaborative agents. To cope with this situation, the Centralized Training with Decentralized Execution (CTDE) framework was put forward. The MARL algorithms embrace CTDE to train the agents' policies with access to the global state, mitigating the non-stationarity and unstable training. Meanwhile, decentralized execution enables scalability by depending on local information. Within this framework, value-decomposition (VD) algorithms estimate the true joint action value function by factoring it into decentralized utility functions and satisfying the Individual-Global-Max (IGM) principle. Additionally, IGM ensures consistency between optimal local and global actions. The core idea is that If the factored versions are accurate, the optimal joint actions can be determined through local greedy operations. However, a particular type of VD approach faces representation limitations by restricting its mixing network parameters to be positive to meet the IGM condition, which simplifies learning but can hinder finding the optimal joint actions.

Another type of VD method endeavors to construct a surrogate target by leveraging the local optimal actions of distributed agents as prior conditions, and narrows the disparity between the surrogate objective and the true joint action-value function to overcome representation limitations. However, in more complex cooperative tasks, this gap may never be fully bridged. It is challenging for the model to simultaneously learn and estimate the true and surrogate joint action-value functions during dynamic interactions between agents, requiring careful calibration of the update weights for both objectives, or incorporating additional external features to enhance the accuracy of joint value function estimation. These limitations result in VD methods of this type either requiring more complex designs or encountering problems like low training efficiency and poor agent collaboration.

To overcome the aforementioned limitations, this paper proposes a Heterogeneous Policy Fusion (HPF) scheme to integrate the benefits of different types of VD methods. Our work shares the same objective of effectively enhancing the training efficiency of MARL algorithms. However, we focus on constructing a novel composite policy to interact with the environment and gather more valuable data instead of developing a new VD scheme. Specifically, HPF expands different VD methods into a policy set with an adaptive composition based on the value estimation, which determines the degree of participation of each candidate policy. To prevent unexpected exploration and suboptimal non-cooperation, an instructive constraint based on KL divergence is imposed between the utility functions of the candidate VD methods. This constraint prevents excessive divergence between their policies, helping to correct the misleading updates in the model. Our contributions can be summed up as threefold:

\begin{itemize}
    \item we analyze the connections and differences among different types of VD methods, and highlight the value of properly integrating these methods in order to enjoy the best of each part, a perspective that is unique and orthogonal to developing a completely new VD scheme.
    \item we propose a simple VD policy extension scheme to leverage the benefits of heterogeneous VD policies, and show that HPF can effectively enhance the training efficiency of the collaboration model.
    \item we verify the effectiveness of HPF by comparing it against different VD methods across several cooperative multi-agent tasks, to demonstrate its superiority to tackle issues like \textit{Relative Overgeneralization}, showcasing its remarkable performance.
\end{itemize}


\section{Related Works}\label{}
Value decomposition is effective in addressing challenges such as partial observability, algorithmic instability, and credit assignment in cooperative environments.
Our focus lies in integrating the advantages of VD methods based on the Individual-Global-Maximum (IGM)~\cite{son2019qtran} principle. Categorized by the different ways of satisfying IGM, the methods fall into two groups: one constrained by network structure, the other constructing a surrogate target and involving information replenishment.

The first category of VD methods decomposes the joint action-value function by imposing constraints on the network structure, such as \textit{additivity} or \textit{monotonicity}. Following the IGM principle, VDN~\cite{sunehag2017value} decomposes the joint value function $Q_{tot}$ into the sum of individual value functions $\sum [Q_{i}]_{i=1}^{n}$, neglecting additional information available during training. While QMIX~\cite{rashid_qmix_2018} imposes monotonicity constraints on the mixing network, making the $Q_{tot}$ and $Q_{i}$ satisfy the IGM principle.
Qatten~\cite{yang2020qatten} theoretically derives the process of decomposing the joint Q-value ($Q_{tot}$) into local Q-values ($Q_{i}$).
VGN~\cite{wei2022vgn} takes into account the influence of agents with the minimum Q-value and utilizes graph attention networks to introduce dynamic relationships between agents.
GraphMIX~\cite{naderializadeh2020graph} introduces graph network and attention mechanisms to the QMIX, integrating the information of agents through the edge weights controlled by attentional relationships between agents.
TransMix~\cite{khan2022transformer} propose a Transformer-based mixture network, emphasizing the extraction of global and local contextual interactions among $Q_i$, historical trajectories, and global states.
The recent MDQ~\cite{ding2023multi} addresses cooperative, competitive, or mixed tasks by combining mean-field theory with VD.
Inspired by the human nervous system, HAVEN~\cite{xu2023haven} devised a two-level QMIX strategy for inter-level and intra-level coordination.
Restricting the network parameters to achieve IGM can lead to representation limitation of the joint action-value function, making it hard to cover the optimal joint actions.

Another category of VD methods satisfies IGM by constructing a surrogate target and supplementing additional information to approximate the true joint action-value function.
QTRAN~\cite{son2019qtran} introduces a soft regularization constraint, wherein the theoretically learned overall reward function $Q_{tran}$ is required to be equal to the true value of the overall reward function $Q_{tot}$ only when the optimal action $\Bar{u}$ is taken. 
While QTRAN++~\cite{son2020qtran++} explicitly specifies the relationship $Q_{tot}(s,\tau,\Bar{u})=Q_{tran}(s,\tau,\Bar{u})>Q_{tran}(s,\tau,u)>Q_{tot}(s,\tau,u)$ and introduces the corresponding loss for $Q_{tot}$ to expedite training. WQMIX~\cite{rashid2020weighted} uses weighted projection to improve QMIX and reduce the learning weight of non-optimal joint actions. QPLEX~\cite{wang2021qplex} converts the IGM principle into consistency constraints for the advantage function, and uses a dual-competitive architecture to factorize $Q_{tot}$. ResQ~\cite{shen2022resq} uses the idea of residual decomposition to decompose the joint Q-value into the sum of a main function and a residual function, to solve the problem of limited representation.
RQN~\cite{pina2022residual} learns an individualized factor for each agent signifying the relative importance of a specific Q-value trajectory.
LOMAQ~\cite{zohar2022locality} proposes a scalable VD method to improve credit assignment by utilizing the learned local agent rewards. 
CIA~\cite{liu2023contrastive} introduces a new contrastive identity-aware learning method to explicitly encourage differentiation between credit levels of agents.
MARGIN~\cite{ding2023multiagent} decomposes the global mutual information into a weighted sum of local mutual information, making it seamlessly integrated with various VD approaches.

Although these methods can obtain more information from the environment or other agents, they usually require a carefully designed structure.
Orthogonal to developing a new VD method, our proposed HPF amalgamates the strengths of the aforementioned VD approaches, aiming to achieve a highly sample-efficient and expressive VD approach.

\section{Preliminaries}\label{pre}
\subsection{Dec-POMDPs}
The formalism known as \textit{Decentralized Partially Observable Markov Decision Processes} (Dec-POMDPs)~\cite{oliehoek_concise_2016}, characterized by the tuple $G=(S, U, P, r, Z, O, n, \gamma)$, is widely employed in delineating the dynamics of complete cooperative multi-agent endeavors. The system comprises $n$ agents, denoted as $n \in N \equiv{1, \ldots, n}$, and the genuine environmental state $s \in S$ encapsulates both the agents' collective knowledge and additional contextual features. At every timestep $t$, the environment transitions based on the state transition function $P\left(s^{\prime} \mid s, \mathbf{u}\right): S \times \mathbf{U} \times S \rightarrow[0,1]$, where $\mathbf{u} \in \mathbf{U} \equiv U^{n}$ encompasses the collectively generated actions with an action $u_i \in U$ chosen by each agent. A global reward function $r(s, \mathbf{u}): S \times \mathbf{U} \rightarrow \mathbf{R}$ is shared among all agents, and each agent receives a partial observation $z^i \in Z$ based on the observation function $O(z^i|s, u^i) : S \times U \rightarrow Z$. The primary goal for all agents is to jointly optimize the team reward $R_{t}=\sum_{i=0}^{T} \gamma^{i} r_{t+i}$, elucidated through the joint value function $Q^{\boldsymbol{\pi}}\left(s_{t}, \mathbf{u}_{t}\right) = \mathbb{E}_{s_{t+1: \infty}, \mathbf{u}_{t+1: \infty}}\left[R_{t} \mid s_{t}, \mathbf{u}_{t}\right]$.

\subsection{Value Function Decomposition}
\subsubsection{Implementation by network constraint}
The connotation of value function decomposition is to establish a direct training connection between the joint state-action value function and the utility functions of the individual agents during centralized training. Thus, for a joint state-action value function $Q_{tot}$, there exist utility functions $[Q_{i}]_{i=1}^{n}, i \in N$ for each of the agents that satisfy:
\begin{equation}
    \begin{aligned}
        & \underset{\mathbf{u}}{\arg \max } Q_{tot} (\boldsymbol{\tau}, \mathbf{u}) = \\
        & (\arg \max _{u^1} Q_1 \left(\tau^1, u^1\right) \ldots \arg \max _{u^n} Q_n\left(\tau^n, u^n\right)).
    \end{aligned}
\end{equation}


This equation can be described as $[Q_{i}]_{i=1}^{n}$ satisfying the IGM criterion for the hypothetical decomposition-available joint-valued function $Q_{tot}$. Therefore, it is also asserted that $Q_{tot}(\boldsymbol{\tau}, \mathbf{u})$ can be decomposed into $[Q_{i}(\tau_{i}, u_{i})]_{i=1}^{n}$ or $[Q_{i}]_{i=1}^{n}$ is a factor of the decomposition of $Q_{tot}$. Two most typical value function decomposition methods are VDN~\cite{sunehag2017value} and QMIX~\cite{rashid_qmix_2018}, which satisfy the sufficient conditions as $Q_{tot}^{VDN}(\boldsymbol{\tau}, \mathbf{u})=\sum_{i=1}^n Q_i\left(\tau^i, u^i\right)$ and $Q_{tot}^{QMIX}(\boldsymbol{\tau}, \mathbf{u})=\sum_{i=1}^n |w_i| Q_i\left(\tau^i, u^i\right)$.
These two methods calculate the global action value function by introducing the mixing network as $Q_{tot} = g(Q_1, Q_2, \cdots, Q_n)$, or factorize the centralized value function with the assurance of monotonicity property, with parameters restricted to be positive. However, the monotonicity constraint limits the representational capacity of the factored global action-value function, which may lead to the agents failing to find the correct optimal joint actions.

\subsubsection{Implementation by constructing surrogate target}
Another type of VD method requires constructing a surrogate objective to approximate the true joint action-value function, ensuring they are equal when taking the optimal joint action. Here, we denote the true joint action-value function as $Q_{jt}$, and the surrogate objective is estimated by the aforementioned VD methods with network constraints as $Q_{tot}$. WQMIX~\cite{rashid2020weighted} tries to pay more attention to the optimal joint action during the training process, and designs to update the true joint action value function as $Q_{jt} \leftarrow w(s,\mathbf{u}) Q_{jt}\left(s, \boldsymbol{\tau}, \operatorname{argmax}_{\mathbf{u}} Q_{t o t}\left(\boldsymbol{\tau}, \mathbf{u}, s\right)\right)$ with weights $w(s,\mathbf{u})$ determined by whether the current joint actions achieves the optimal one, which is depicted as:

\begin{equation}
    w(s, \mathbf{u})= \begin{cases}1 & \mathbf{u}=\mathbf{u}^*=\operatorname{argmax}_{\mathbf{u}} Q_{jt}(s, \mathbf{u}) \\ & or \ Q_{t o t}(s, \mathbf{u})<Q_{jt}(s, \mathbf{u}), \\ \alpha & \text { otherwise. }\end{cases}
\end{equation}

While QPLEX designs an Advantage-based IGM principle, which loosens the expressiveness of the IGM-class factored global action value function through a duplex dueling architecture as:

\begin{equation}
    \begin{aligned}
        & Q_{t o t}(\boldsymbol{\tau}, \mathbf{u})=V_{t o t}(\boldsymbol{\tau})+A_{t o t}(\boldsymbol{\tau}, \mathbf{u}) \\
        & with \  V_{t o t}(\boldsymbol{\tau})=\max _{\mathbf{u}^{\prime}} Q_{t o t}\left(\boldsymbol{\tau}, \mathbf{u}^{\prime}\right).
    \end{aligned}
\end{equation}

\begin{figure*}[t]
    \centering
    \includegraphics[width=0.95\linewidth]{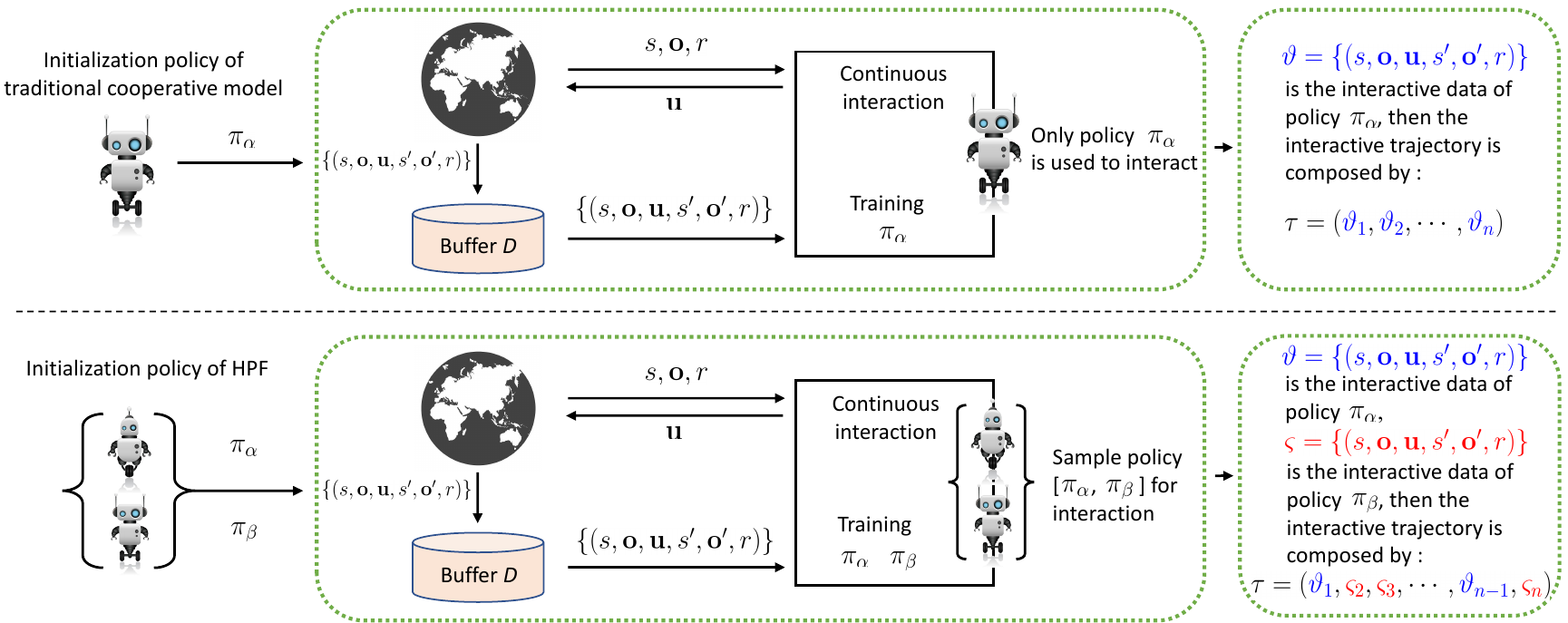}
    \caption{{\bf The illustration of the distinction between HPF and traditional VD methods.} The traditional scheme directly aligns the optimal joint action and optimizes the central value function with the presupposed VD policy itself. The proposed HPF integrates the benefits of different types of VD policies, and expands them into a policy set to sample the experiences for capturing further performance improvement. These VD policies both participate in the interactions with the environment and learning in an adaptive manner.}
    \label{hpf_frame}
\end{figure*}

\section{Methodology}\label{}
In this section, we will introduce the policy fusion pattern HPF in detail. Specifically, we begin by introducing the motivation behind our proposed scheme and then delve into the in-depth analysis of the composite policy in HPF. Additionally, we design an instructive constraint between the utility functions of composite policy, which aims to correct the learned factorized utilities and provide the VD learner with a more accurate estimation potential. Finally, we depict the overall implementation of HPF at the end of this section.

\subsection{Motivation: connections and gaps between different VD policies}
Our proposed fusion framework is founded on the VD paradigm. It decomposes the centralized value function into a combinatorial form of the utility functions for the collaborative model to learn. The aim of this factorization is to align the optimal actions of the centralized value function with the local optimal actions of individual agents. Then the collaborative model can update the utility functions by maximizing the $Q$-value while meeting the IGM criterion.

As summarized in Section \ref{pre}, even though the methods of the two existing types of VD policies are both intended to more effectively align the optimal joint actions of the centralized value function and the local optimal actions of each agent, there are still obvious commonalities and differences among various approaches. These relationships and distinctions can be summed up as the following points.

\subsubsection{For VD by network constraint}
\begin{itemize}
    \item {\it Direct centralized value function decomposition:} it requires fewer neural network components to be optimized, resulting in a faster optimization speed and higher efficiency.
    \item {\it Hard constraints on mixing network parameters:} the assurance of the IGM criterion is realized by limiting the parameters and structure of the mixing network.
    \item {\it Limited representation capacity:} the restrictions of the mixing network prevent the correct fitting of the negative correlation between the centralized value function and the utility functions.
\end{itemize}

\subsubsection{For VD by constructing surrogate target}
\begin{itemize}
    \item {\it Surrogate target reservation:} the constructed surrogate target is incorporated as a new element into the original VD procedure, which may bring about relatively low training efficiency.
    \item {\it Soft constraints on optimization:} both the optimization and alignment process of actions are softly constrained by {\it MSE loss}, which may make it difficult to train the model with multiple objectives to be optimized.
    \item {\it Full representation capacity:} the surrogate target secludes the original true centralized value function objective, which maintains the full representation capacity of the neural network, enabling it to fully fit the complicated reward distribution.
\end{itemize}

Besides the optimization process and algorithmic connections, the two types of VD policies are also complementary to each other in terms of strengths and weaknesses. The VD policies with network constraints are more efficient since they have fewer components to optimize, but they may not be able to find the true optimal joint action with the limited representation. While another type of method enjoys the potential to find the true optimal joint action, but is comparatively much less sample efficient. This inevitably poses a question:

{\it ``Is there a way to integrate the benefits of the two types of VD policies with the interactions between the agents and environments?"}

Because of the connections and complementary strengths above, instead of treating the different types of policies as two isolated aspects, it is more natural to connect both in pursuit of a performant policy in practice.

\begin{figure*}[t]
    \centering
    \includegraphics[width=0.85\linewidth]{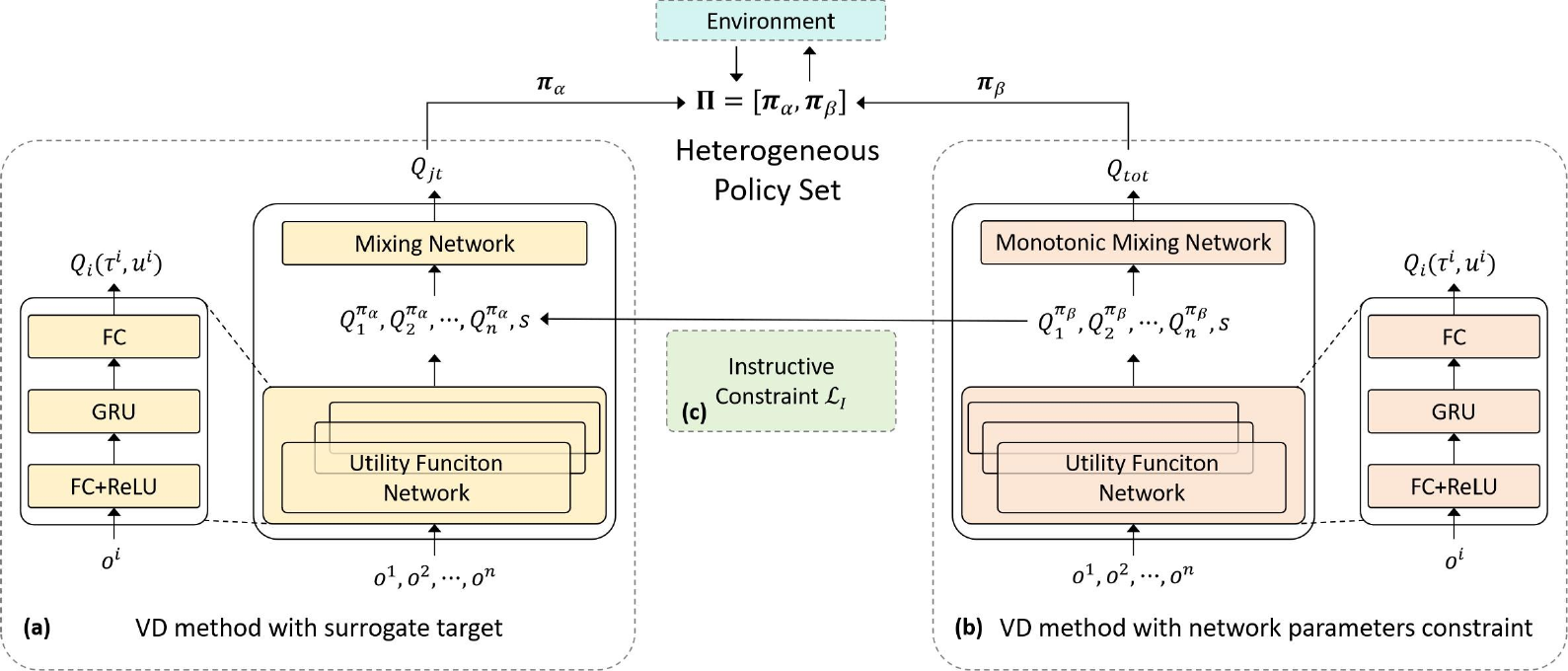}
    \caption{{\bf The architecture of HPF.} (a) The VD method with surrogate target. (b) The VD method with network parameters constraint. (c) The instructive constraint between heterogeneous utility functions. The policies of both VD methods constitute a composite policy set and interact with the environment after sampling.}
    \label{learning-process}
\end{figure*}

\subsection{Policy Extension and Adaptive Composition}\label{hpf}
\noindent {\bf Policy Extension.}
The policy extension in HPF is direct and readily combined with existing VD methods, which is illustrated in Figure \ref{hpf_frame}. Before the collaborative model interacts with the environment, we consider the representation-limited type of VD policy as a candidate policy. It is then extended into the surrogate-target type of policy, forming a candidate set. Subsequently, the collaborative model adaptively synthesizes a composite policy from this set to interact with the environment, which results in the mixed sample trajectory simultaneously containing interaction experiences derived from the two types of VD approaches. Since the mixed sample trajectory potentially includes policy information from two different VD policies, both of them would preserve their respective characteristics while improving the shortcomings of each when they sample these experiences to train the collaborative model. The expanded policy set $\Pi$ is depicted as:
\begin{equation}\label{extended_pi}
    \Pi=\left[\boldsymbol{\pi}_\alpha, \boldsymbol{\pi}_\beta \right],
\end{equation}

\noindent where $\boldsymbol{\pi}_\alpha$ represents the VD policy with surrogate target, and $\boldsymbol{\pi}_\beta$ is VD policy with limited representation. In the context of the extended set of policies formed by existing VD policies, the model continues to meet the IGM criterion, which is shown in Proposition \ref{prop:extended_pi}. The policies in $\Pi=\left[\boldsymbol{\pi}_\alpha, \boldsymbol{\pi}_\beta \right]$ will all get involved in the interactions and learning in a collaborative and adaptive manner as detailed in the next part.

\begin{proposition}\label{prop:extended_pi}
    If $\Pi=\left[\boldsymbol{\pi}_\alpha, \boldsymbol{\pi}_\beta \right]$ is an extended policy set formed by existing various VD policies, then $\Pi$ still satisfies the IGM criterion.
\end{proposition}

\begin{proof}
    See Section A of Supplementary Materials.
\end{proof}

\noindent {\bf Adaptive Composition.}
Both VD policies in $\Pi$ will form a single composite policy $\boldsymbol{\tilde{\pi}}$ from the heterogeneous VD policies in $\Pi$. More specifically, given the current state $s$ and the joint observations $\boldsymbol{o}$, HPF first samples actions for each member of the policy set $\Pi$ and then calculates the utility functions of each candidate VD policy. Since these utilities will contribute to determining the joint action-value functions, which is the lead of VD policy to update, we use them as a proxy to composite the final policy with probability based on them. For example, we estimate their values at the current state as $\mathbb{Q} = \{ \boldsymbol{Q}^{\boldsymbol{\pi}_\alpha}, \boldsymbol{Q}^{\boldsymbol{\pi}_\beta} \}$. While each agent chooses the action with $\epsilon$-greedy as the interactive action, the set of utility functions is represented as $\boldsymbol{Q}^{\boldsymbol{\pi}_\alpha} = \{Q^{\boldsymbol{\pi}_\alpha}_{1}, Q^{\boldsymbol{\pi}_\alpha}_{2}, \ldots, Q^{\boldsymbol{\pi}_\alpha}_{n} \}$ and $\boldsymbol{Q}^{\boldsymbol{\pi}_\beta} = \{Q^{\boldsymbol{\pi}_\beta}_{1}, Q^{\boldsymbol{\pi}_\beta}_{2}, \ldots, Q^{\boldsymbol{\pi}_\beta}_{n} \}$. Subsequently, HPF then constructs a Boltzmann policy-based categorical distribution for selecting the final interactive policy:

\begin{equation}\label{pi}
    P_{\mathbf{w}}[k]=\frac{\exp \left(\boldsymbol{Q}^{\boldsymbol{\pi}_k} / \eta\right)}{\sum_k \exp \left(\boldsymbol{Q}^{\boldsymbol{\pi}_k} / \eta\right)}, \quad \forall k \in[\alpha, \beta]
\end{equation}

\noindent where $\eta$ is the temperature of the Boltzmann distribution. Then we sample from $\mathbf{w} \sim P_{\mathbf{w}}$ for integrating the composite policy to determine which policy will be implemented during unroll in interaction. The composite policy $\boldsymbol{\tilde{\pi}}$ can be conceptually represented in the following manner:

\begin{equation}\label{pi2}
    \boldsymbol{\tilde{\pi}}(\mathbf{u} \mid \boldsymbol{o})=\left[\delta_{\mathbf{u} \sim \boldsymbol{\pi}_\alpha}, \delta_{\mathbf{u} \sim \boldsymbol{\pi}_\beta}\right] \cdot \mathbf{w}, \quad \mathbf{w} \sim P_{\mathbf{w}}
\end{equation}

\noindent where $\mathbf{w}$ is a one-hot vector, indicating the policy that is selected for the joint observation $\boldsymbol{o}$. $\delta_{\mathbf{u} \sim \boldsymbol{\pi}}$ denotes the Dirac delta distribution centered at $\mathbf{u}$ which is sampled from $\boldsymbol{\pi}$.

The estimation method of $\boldsymbol{Q}^{\boldsymbol{\pi}_k}$ in Equation \eqref{pi} is crucial in the computation of $P_{\mathbf{w}}$ because it influences the sampling process and fusion level of different heterogeneous VD policies within the $\Pi$. Therefore, we design two principles to estimate $\boldsymbol{Q}^{\boldsymbol{\pi}_k}$: (1) it uses the sum of utility functions in $\Pi$ as the sampling criterion inspired by VDN~\cite{sunehag2017value}; (2) it uses the estimation of joint optimal action value functions for the sampling, which takes into account the differences in the representational capabilities~\cite{rashid2020weighted}. The detailed calculation is outlined as follows:

\begin{equation}\label{version}
    \begin{aligned}
        & \text{(Additive)} \ \boldsymbol{Q}^{\boldsymbol{\pi}_k} = \sum_{i=1}^{n}Q_{i}^{\boldsymbol{\pi}_{k}} \\
        & \text{(Optimistic)} \ \boldsymbol{Q}^{\boldsymbol{\pi}_k} = Q_{jt}^{\boldsymbol{\pi}_{k}}(\boldsymbol{\tau}, \hat{\mathbf{u}}^*, s),
    \end{aligned}
\end{equation}

\noindent where $k \in [\alpha, \beta]$, $i \in \{1, \ldots, n\}$ indicate the agent ID, and $\hat{\mathbf{u}}^*=\operatorname{argmax}_{\mathbf{u}} Q_{tot}(\boldsymbol{\tau}, \mathbf{u}, s)$ means the set of maximal local actions. After that, the model interacts with the environment and forms a mixed trajectory to store in the replay buffer. During the training phase, the model samples a batch of sampled trajectories and trains both types of VD policies in $\Pi$ simultaneously. Since the chosen interaction policy at each time step may differ, these trajectories may contain segments of experiences from different types of policies, intertwining the mixed experiences throughout the entire trajectory.

In general, HPF shares the same core as VD methods, aiming to enhance the training efficiency of collaborative policy. Orthogonal to designing a new factorization scheme, HPF focuses on integrating the advantages of existing VD methods using a policy fusion scheme. This also leverages the mixed experiences to guide and promote mutual interactions between the candidate VD policies in $\Pi$. In comparison to existing VD methods, HPF possesses several advantages:
\begin{itemize}
    \item \textbf{Independence of candidate interaction policies:} The policies in the set $\Pi=\left[\boldsymbol{\pi}_\alpha, \boldsymbol{\pi}_\beta \right]$ are mutually independent. Since the composite policy is implemented through sampling, the chosen interaction policy by the model remains unaffected by other policies in $\Pi$.
    \item \textbf{Flexibility in the form of interaction policies:} The policies in the set $\Pi$ vary in type and network architecture. Since HPF aims to integrate the advantages of diverse policies, then $\boldsymbol{\pi}_{\alpha}$ and $\boldsymbol{\pi}_{\beta}$ are assigned two distinct types of policies within the VD framework. In a broader sense, HPF can flexibly adopt any method, allowing policies $\boldsymbol{\pi}_{\alpha}$ and $\boldsymbol{\pi}_{\beta}$ to differ in form, algorithm, or category based on the user’s integration needs.
    \item \textbf{Adaptive policy selection:} Every policy within the set $\Pi$ is available for interaction with the environment, engaging adaptively.
\end{itemize}

\subsection{Instructive Constraint}
Despite policy sampling, HPF's composite policy effectively leverages heterogeneous VD methods. However, VD methods in $\Pi$ enforcing the IGM criterion via parameter constraints still face representational limitations and often underperform due to inaccurate sampling of optimal joint actions during training. To address this, we introduce an instructive constraint to guide agents among the VD policies in the extended set $\Pi$:

\begin{equation}\label{instructive}
    \mathcal{L}_I=\min D_{\mathrm{KL}}\left(\pi_{\alpha}^i\left(\cdot \mid \tau^i\right) \| \pi_{\beta}^i(\cdot \mid \tau^i)\right),
\end{equation}

\noindent where $\pi^i\left(\cdot \mid \tau^i\right)$ is a Boltzmann policy with respect to each
agent’s decentralized utility function is defined as:

\begin{equation}
    \pi_{k}^i\left(u^i \mid \tau^i\right)=\frac{\exp \left(Q_i\left(\tau^i, u^i\right)\right)}{\sum_{u^i} \exp \left(Q_i\left(\tau^i, u^i\right)\right)}, \forall u^i \in U^i, k \in [\alpha, \beta].
\end{equation}

The intuition behind this constraint is that VD methods constructing a surrogate target can learn the correct optimal joint action function, thereby identifying the optimal local behaviors. By making decentralized policies in the representationally constrained VD methods imitate this, agents are encouraged to select local optimal actions aligned with the correct joint action. This enables both VD learners to approximate the optimal joint action while adhering to the IGM principle, promoting efficient collaborative training.

\subsection{Overall Implementation}
As shown in Figure \ref{learning-process}, HPF consists of three parts: the joint action value function $Q_{jt}$ estimated by the VD method with complete representation, the factored global action value function $Q_{tot}$ and the instructive constraint $\mathcal{L}_{I}$ between the utility functions of both methods. All modules are implemented by neural networks, with parameters shared across agents to promote efficient policy learning. The overall learning objective of HPF is defined as follows:

\begin{equation}\label{objective}
    \mathcal{L}=\mathcal{L}_{\mathrm{TD}}^{tot}+\mathcal{L}_{\mathrm{TD}}^{jt}+\mathcal{L}_I.
\end{equation}

In Equation \eqref{objective}, $\mathcal{L}_{\mathrm{TD}}^{tot}$ and $\mathcal{L}_{\mathrm{TD}}^{jt}$ respectively denote the update losses of the distinct heterogeneous factored global action-value function, given by:

\begin{equation}
    \mathcal{L}_{\mathrm{TD}}^{tot}=\sum_{1}^b\left(Q_{tot}\left(s, \mathbf{u}\right)-y_{tot}^{\prime}\right)^2,
\end{equation}

\begin{equation}
    \mathcal{L}_{\mathrm{TD}}^{jt}=\sum_{1}^b\left(Q_{jt}\left(s, \mathbf{u}\right)-y_{jt}^{\prime}\right)^2,
\end{equation}

\noindent where $y_{tot}^{\prime}=r+\gamma Q_{tot}\left(s^{\prime}, \operatorname{argmax}_{\mathbf{u}^{\prime}} Q_{{tot}}\left(\boldsymbol{\tau}^{\prime}, \mathbf{u}^{\prime}\right)\right)$ and $y_{jt}^{\prime}=r+\gamma Q_{jt}\left(s^{\prime}, \operatorname{argmax}_{\mathbf{u}^{\prime}} Q_{{tot}}\left(\boldsymbol{\tau}^{\prime}, \mathbf{u}^{\prime}\right)\right)$ is the respective learning target for both VD policy and $b$ are the batches sampled by $\Pi = [\boldsymbol{\pi}_{\alpha}, \boldsymbol{\pi}_{\beta}]$. $\mathcal{L}_{I}$ denotes the instructive constraint defined in Equation \eqref{instructive}. In addition, we utilize the target networks to stabilize the training. More details of the HPF can be found in Section B of Supplementary Materials.

\section{Experiments}\label{}
In this section, we constructed two sets of HPF composite policies (WQMIX and QMIX, QPLEX and VDN, referred to as HPF-WQ and HPF-QV). Each composite policy utilizes two evaluation methods from Equation \eqref{version} for policy sampling in more difficult tasks (hence denoted as Add-HPF-WQ, Opt-HPF-WQ, Add-HPF-QV, and Opt-HPF-QV). We then compared these methods with state-of-the-art baselines: WQMIX, QPLEX, QMIX, VDN, and ResQ. For a fair comparison, all experimental results are illustrated with the median performance and the 25\%-75\% quartile over 5 random seeds. More details about hyperparameters and experiments are provided in Section C of Supplementary Materials.

\newcolumntype{P}[1]{>{\centering\arraybackslash}p{#1}}
\begin{table}[t]
    \scriptsize
        \setlength{\tabcolsep}{-1pt}
        \begin{minipage}{.49\columnwidth}
          \centering
            \begin{tabular}{|P{1cm}||P{1cm}|P{1cm}|P{1cm}|}
            \hline
            \backslashbox{$U_1$}{$U_2$}& $u_1$ & $u_2$ & $u_3$\\ \hline \hline
            $u_1$ & \textbf{8}& -12 & -12 \\ \hline
            $u_2$ & -12 & 3 & 0 \\ \hline
            $u_3$ & -12 & 0 & 3 \\ \hline
            \end{tabular}
            \subcaption{Payoff matrix}
            \label{table:matrix1}
        \end{minipage}
        \begin{minipage}{.49\columnwidth}
          \centering
            \begin{tabular}{|P{1cm}||P{1cm}|P{1cm}|P{1cm}|}
            \hline
            \backslashbox{$U_1$}{$U_2$} & $u_1$ & $u_3$ & $u_3$\\ \hline \hline
            $u_1$ & -6.91 & -4.71 & -4.62 \\ \hline
            $u_2$ & -4.59 & -2.41 & \textbf{-2.33} \\ \hline
            $u_3$ & -4.72 & -2.52 & -2.46 \\ \hline
            \end{tabular}
            \subcaption{Payoff learned by VDN}
            \label{table:matrix2}
        \end{minipage} \\
        \begin{minipage}{.49\columnwidth}
            \centering
              \begin{tabular}{|P{1cm}||P{1cm}|P{1cm}|P{1cm}|}
              \hline
              \backslashbox{$U_1$}{$U_2$}& $u_1$ & $u_2$ & $u_3$\\ \hline \hline
              $u_1$ & -8.05 & -8.05 & -8.06 \\ \hline
              $u_2$ & -8.06 & \textbf{1.49} & 1.42 \\ \hline
              $u_3$ & -8.06 & 1.44 & 1.37 \\ \hline
              \end{tabular}
              \subcaption{Payoff learned by QMIX}
              \label{table:matrix3}
          \end{minipage}
          \begin{minipage}{.49\columnwidth}
            \centering
              \begin{tabular}{|P{1cm}||P{1cm}|P{1cm}|P{1cm}|}
              \hline
              \backslashbox{$U_1$}{$U_2$} & $u_1$ & $u_3$ & $u_3$\\ \hline \hline
              $u_1$ & \textbf{9.66} & -12.78 & -12.76 \\ \hline
              $u_2$ & -12.23 & 3.16 & -0.75 \\ \hline
              $u_3$ & -11.54 & -0.70 & 3.12 \\ \hline
              \end{tabular}
              \subcaption{Payoff learned by QPLEX}
              \label{table:matrix4}
          \end{minipage} \\
          \begin{minipage}{.49\columnwidth}
            \centering
              \begin{tabular}{|P{1cm}||P{1cm}|P{1cm}|P{1cm}|}
              \hline
              \backslashbox{$U_1$}{$U_2$}& $u_1$ & $u_2$ & $u_3$\\ \hline \hline
              $u_1$ & \textbf{8.03}& -11.98 & -12.02 \\ \hline
              $u_2$ & -11.98 & 1.53 & 1.54 \\ \hline
              $u_3$ & -12.09 & 1.54 & 1.53 \\ \hline
              \end{tabular}
              \subcaption{Payoff learned by WQMIX}
              \label{table:matrix5}
          \end{minipage}
          \begin{minipage}{.49\columnwidth}
            \centering
              \begin{tabular}{|P{1cm}||P{1cm}|P{1cm}|P{1cm}|}
              \hline
              \backslashbox{$U_1$}{$U_2$} & $u_1$ & $u_3$ & $u_3$\\ \hline \hline
              $u_1$ & -0.70 & -0.24 & -0.16 \\ \hline
              $u_2$ & -0.26 & 2.01 & 2.52 \\ \hline
              $u_3$ & -0.19 & 2.60 & \textbf{3.16} \\ \hline
              \end{tabular}
              \subcaption{Payoff learned by ResQ}
              \label{table:matrix6}
          \end{minipage} \\
          \begin{minipage}{.49\columnwidth}
            \centering
              \begin{tabular}{|P{1cm}||P{1cm}|P{1cm}|P{1cm}|}
              \hline
              \backslashbox{$U_1$}{$U_2$}& $u_1$ & $u_2$ & $u_3$\\ \hline \hline
              $u_1$ & \textbf{7.99}& -11.99 & -12.00 \\ \hline
              $u_2$ & -12.02 & 3.00 & -0.01 \\ \hline
              $u_3$ & -12.04 & -0.04 & 2.99 \\ \hline
              \end{tabular}
              \subcaption{Payoff learned by HPF-WQ}
              \label{table:matrix7}
          \end{minipage}
          \begin{minipage}{.49\columnwidth}
            \centering
              \begin{tabular}{|P{1cm}||P{1cm}|P{1cm}|P{1cm}|}
              \hline
              \backslashbox{$U_1$}{$U_2$} & $u_1$ & $u_3$ & $u_3$\\ \hline \hline
              $u_1$ & \textbf{7.93} & -12.56 & -12.35 \\ \hline
              $u_2$ & -12.32 & 3.02 & 0.05 \\ \hline
              $u_3$ & -10.67 & -0.06 & 2.97 \\ \hline
              \end{tabular}
              \subcaption{Payoff learned by HPF-QV}
              \label{table:matrix8}
          \end{minipage} \\
        \caption{Pay-off matrix of one-step game. Boldface means optimal/greedy actions from the state-action value.}
    \label{table:matrix}
\end{table}

\subsection{Matrix Game}
Our evaluation begins with a matrix game shown in Table \ref{table:matrix1}, where two agents need to perform the optimal joint action $(u_1, u_1)$ to obtain the highest reward. However, sub-optimal joint actions $(u_2, u_2)$ and $(u_3, u_3)$ tend to make decentralized agents choose their local actions $u_2$ or $u_3$ when other agents select actions uniformly during training, which will lead to the noted \textit{Relative-Overallgeneralization} problem.

Table \ref{table:matrix} presents the comparison results of various methods in matrix game. We can observe that QMIX fails to select the optimal joint action and struggles in sub-optimal ones, while VDN even can not accurately approximate the suboptimal joint action values. Although QPLEX achieves success, it tends to overestimate the optimal joint action values. Theoretically, both WQMIX and ResQ can satisfy the IGM condition, but they still cannot correctly estimate the centralized value function distribution on this slightly more complex matrix game, and may even lead to the failure of optimal joint action selection, which may require more training. In contrast, our algorithms (HPF-WQ, HPF-QV) effectively identify and select the optimal joint action, which demonstrates superior performance than other baselines.

\begin{figure*}[t]
    \centering
    \includegraphics[width=1\linewidth]{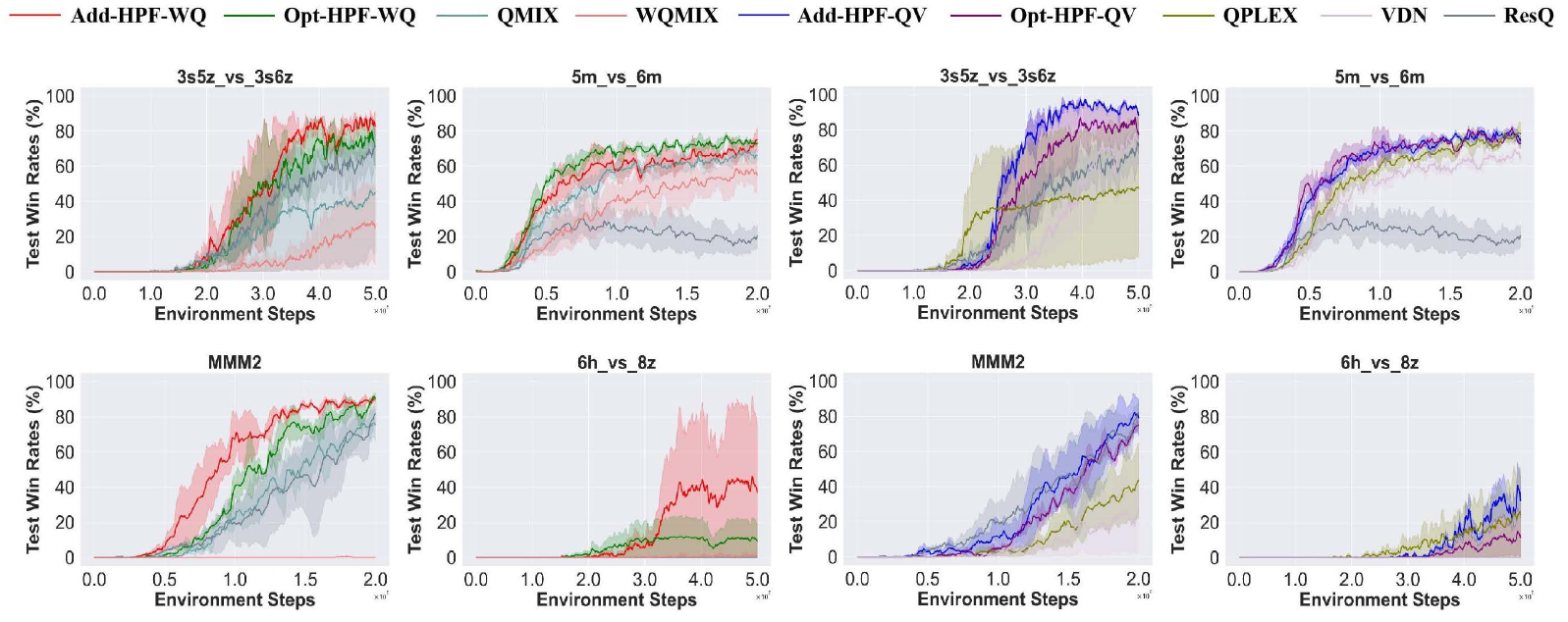}
    \caption{Comparison results on the selected scenarios in the StarCraft Multi-Agent Challenge.}
    \label{smac_results}
\end{figure*}

\begin{figure}[h]
    \centering
    \includegraphics[width=0.49\linewidth]{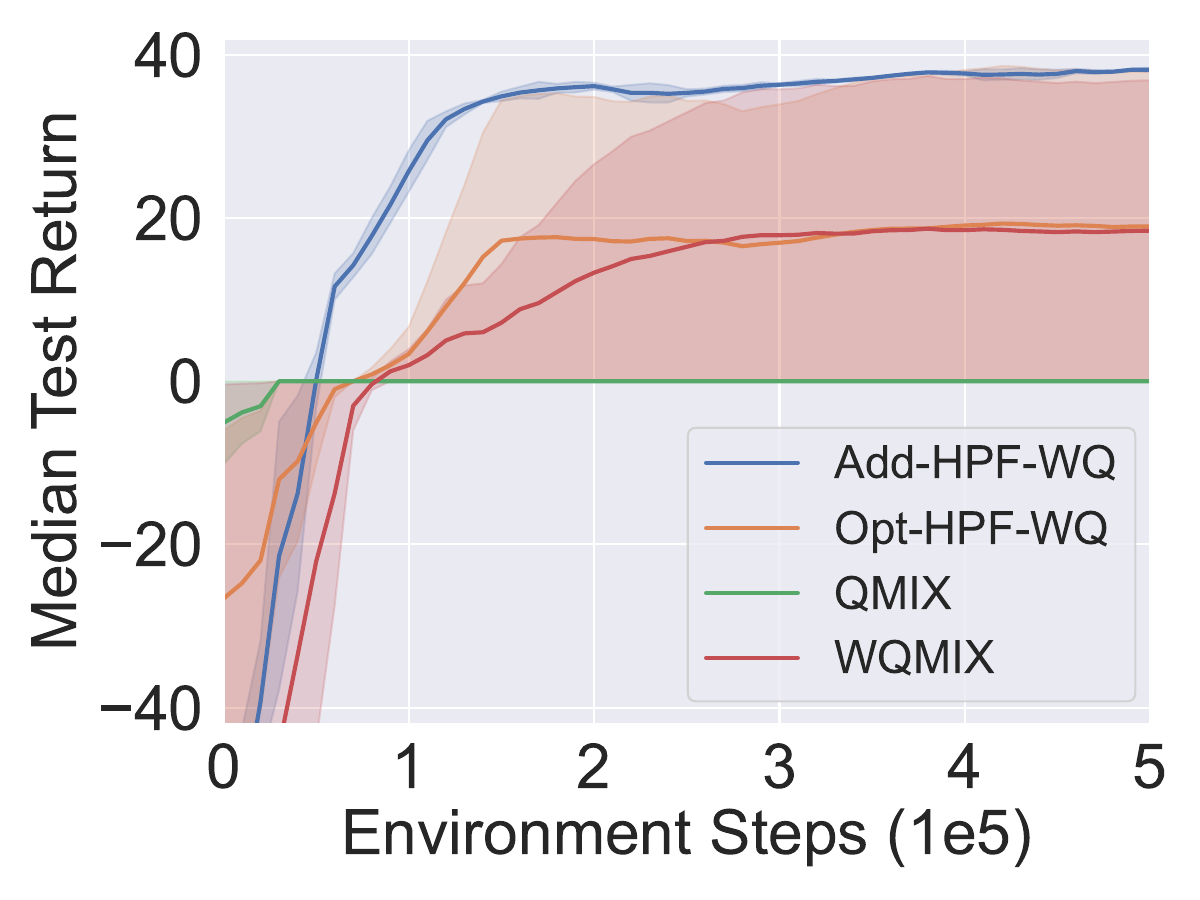}
    \includegraphics[width=0.49\linewidth]{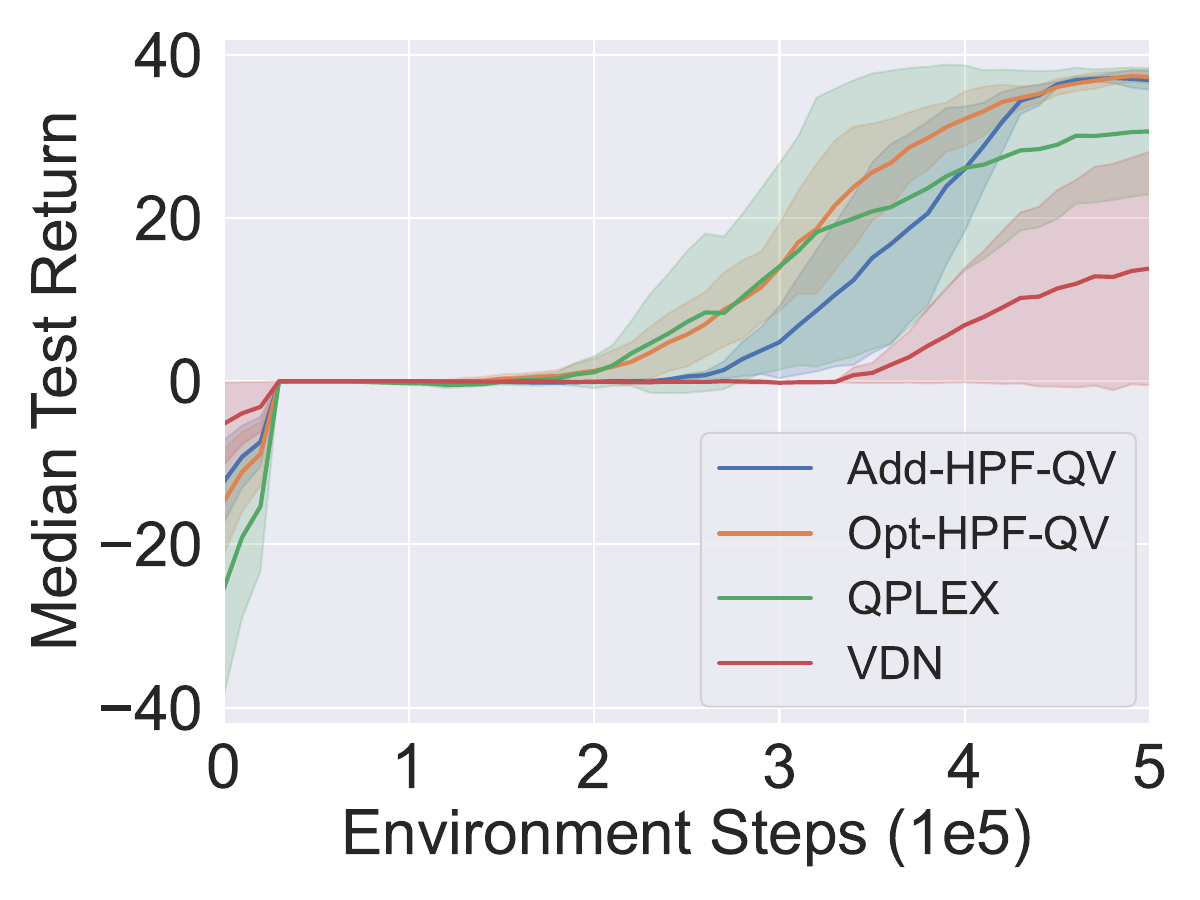}
    \caption{Comparison results in the predator and prey.}
    \label{pp}
\end{figure}

\subsection{StarCraft Multi-Agent Challenge}
To further showcase HPF's scalability in more complex environments, we test it using the StarCraft Multi-Agent Challenge (SMAC) benchmark. The evaluation covers different levels of difficulty and complexity, and focuses on five maps: 3s\_vs\_5z, 5m\_vs\_6m, MMM2, 6h\_vs\_8z and 3s5z\_vs\_3s6z, which serve as our test environments.

Figure \ref{smac_results} shows the performance of all the algorithms on these selected scenarios. One can observe that our algorithm (Add-HPF-WQ, Opt-HPF-WQ, Add-HPF-QV and Opt-HPF-QV) significantly outperforms the counterpart baselines (WQMIX, QMIX, and QPLEX, VDN) and quickly learns a cooperative policy to achieve high testing win rates on all maps. Particularly in the super-hard challenging scenarios (6h\_vs\_8z and 3s5z\_vs\_3s6z), HPF demonstrated impressive capabilities. Though WQMIX can learn the correct action-value function, its poor performance is hard to be noted. The effectiveness of HPF stems from combining the strengths of two types of VD methods, enabling the discovery of beneficial samples during strategy sampling and interaction with the environment, thus enhancing overall performance.

\subsection{Predator and Prey}
We further evaluate the effectiveness of the HPF scheme by conducting experiments in the predator-prey environment, a partially observable task with 8 predators and 8 prey. A negative reward of -2 is given when a single predator attempts to capture prey alone, and a positive reward of +10 is received when two predators successfully cooperate to capture prey.

As shown in Figure \ref{pp}, QMIX and VDN fail to learn effective policies to capture the prey. Although WQMIX benefits from the complete IGM expressiveness, its poor performance and high variance suggest difficulties in approximating the true joint action-value function. QPLEX has some competitiveness, but the effect is still inferior to our proposed scheme, which further demonstrates the effectiveness of HPF.

\subsection{Ablation Studies}
To test the effectiveness of the sampling method for VD policies in HPF, we replaced the sampling approach in Equations \eqref{pi}-\eqref{pi2} with random sampling to validate its effectiveness. The experimental results, as presented in Section D of Supplementary Materials, demonstrate that the composite policy in HPF with sampling based on the estimated value function achieves superior performance. The sampling method in HPF enables the model to identify more suitable sampling policies to interact, thus improving the overall performance of the model.

\section{Conclusion}
This paper introduces the Heterogeneous Policy Fusion scheme for cooperative multi-agent reinforcement learning, which balances representation ability and training efficiency in value decomposition algorithms by combining policies with both restricted and unrestricted representation capacities. HPF designs a composite policy by extending two distinct types of value decomposition policies into a policy set, and integrating the interaction policy based on their value function estimates. This extension enables HPF to attain both high training efficiency and complete representation capacity. Experiments show that HPF not only learns optimal joint actions effectively but also outperforms most baseline VD methods without needing specialized design, proving its effectiveness and ease of use.

\bibliographystyle{named}
\bibliography{ijcai25}


\newpage
\section{Appendix: Brief Introduction}
This supplementary material provides the proof procedure, additional environment settings, and ablation results that were omitted from the main text. In Section A, we provide the proof of Proposition 1. Section B describes the specific implementation process of HPF, which takes WQMIX and QMIX as examples of two candidate VD policies of different types. We describe the experimental scenarios and hyperparameters setup in Section C in detail. The additional analysis and omitted ablation results are presented in Section D.

\appendix
\section{Proof of Proposition 1}\label{app:A}
In this section, we show the proof of Proposition 1.

\begin{proposition}
    If $\Pi=\left[\boldsymbol{\pi}_\alpha, \boldsymbol{\pi}_\beta \right]$ is an extended policy set formed by existing heterogeneous VD policies, then $\Pi$ still satisfies the IGM criterion.
\end{proposition}

\begin{proof}
    The Theorem 1 and Section 4 of ResQ ~\cite{shen2022resq} show that {\it for any hard-to-factorize state-action value function $Q_{jt}$, It can generally always be regarded as satisfying the following formula:}
    \begin{equation}\label{eq:Q_jt}
        Q_{j t}(\boldsymbol{\tau}, \boldsymbol{u})=w_{t o t}(\boldsymbol{\tau}, \boldsymbol{u}) Q_{t o t}(\boldsymbol{\tau}, \boldsymbol{u}) + w_r(\boldsymbol{\tau}, \boldsymbol{u}) Q_r(\boldsymbol{\tau}, \boldsymbol{u})
    \end{equation}

    \noindent where masks $w_{t o t}(\boldsymbol{\tau}, \boldsymbol{u})$, $w_{r}(\boldsymbol{\tau}, \boldsymbol{u}) \in \{0, 1\}$. The main $Q_{tot}$ shares the same greedy optimal policy as $Q_{jt}$, and $Q_{t o t}(\boldsymbol{\tau}, \boldsymbol{u})$ is factorized into per-agent utilities $Q_i$, and each agent selects its own optimal action greedily with respect to $Q_{i}$.

    More generally, the value decomposition policies that construct surrogate targets can all be regarded as special cases of Formula \eqref{eq:Q_jt} and can be represented in a generalized form as:

    \begin{equation}
        Q_{jt}^{Type1}(\boldsymbol{\tau}, \boldsymbol{u}) = Q_{tot}^{(1)}(\boldsymbol{\tau}, \boldsymbol{u})+w_r(\boldsymbol{\tau}, \boldsymbol{u}) Q_r(\boldsymbol{\tau}, \boldsymbol{u})
    \end{equation}

    \noindent where $Q_r(\boldsymbol{\tau}, \boldsymbol{u}) \leq 0$, $Q_{tot}^{(1)}(\boldsymbol{\tau}, \boldsymbol{u})$ and $\left[Q_i\left(\tau_i, u_i\right)\right]_{i=1}^n$ satisfy the monotonicity condition $\partial Q_{tot}^{(1)}(\boldsymbol{\tau}, \boldsymbol{u}) / \partial Q_i\left(\tau_i, u_i\right) \geq 0, \quad \forall i \in N$ and

    \begin{equation}\label{w_r}
        w_r(\boldsymbol{\tau}, \boldsymbol{u})= \begin{cases}0 & \boldsymbol{u}=\overline{\boldsymbol{u}} \\ 1 & \boldsymbol{u} \neq \overline{\boldsymbol{u}}\end{cases}
    \end{equation}

    While the value decomposition policies by network constraints still satisfy the monotonicity and are depicted as:

    \begin{equation}
        Q_{jt}^{Type2}(\boldsymbol{\tau}, \boldsymbol{u}) = Q_{tot}^{(2)}(\boldsymbol{\tau}, \boldsymbol{u}) = \sum_{i=1}^{n}|w_i|Q_i\left(\tau_i, u_i\right)
    \end{equation}

    For the extended composite policy set $\Pi=\left[\boldsymbol{\pi}_\alpha, \boldsymbol{\pi}_\beta \right]$ in the HPF framework, the VD policies it contains are sampled for interaction and update, so its centralized value function $Q_{jt}$ is expressed as:

    \begin{equation}
        \begin{aligned}
            Q_{jt}(\boldsymbol{\tau}, \boldsymbol{u}) & = \lambda Q_{jt}^{Type1}(\boldsymbol{\tau}, \boldsymbol{u}) + (1-\lambda) Q_{jt}^{Type2}(\boldsymbol{\tau}, \boldsymbol{u}) \\
            & = \lambda (Q_{tot}^{(1)}(\boldsymbol{\tau}, \boldsymbol{u})+w_r(\boldsymbol{\tau}, \boldsymbol{u}) Q_r(\boldsymbol{\tau}, \boldsymbol{u})) \\
            & + (1-\lambda) Q_{tot}^{(2)}(\boldsymbol{\tau}, \boldsymbol{u})
        \end{aligned}
    \end{equation}

    \noindent where $Q_r(\boldsymbol{\tau}, \boldsymbol{u}) \leq 0$ and $\lambda \in [0, 1]$ is the sample probability of the candidate VD policy in $\Pi$. We show that $\arg \max _{\boldsymbol{u}} Q_{jt}(\boldsymbol{\tau}, \boldsymbol{u}) = \boldsymbol{\bar{\boldsymbol{u}}}$, where $\bar{u}_i=\arg \max _{u_i} Q_i\left(\tau_i, u_i\right)$ and $\boldsymbol{\bar{\boldsymbol{u}}}=\left[\bar{u}_i\right]_{i=1}^n$.

    Then the inequality holds:
    
    \begin{subequations}\label{igm_hpf}
        \begin{align}
            Q_{jt}(\boldsymbol{\tau}, \bar{\boldsymbol{u}}) & = \lambda Q_{tot}^{(1)}(\boldsymbol{\tau}, \bar{\boldsymbol{u}}) + (1-\lambda) Q_{tot}^{(2)}(\boldsymbol{\tau}, \bar{\boldsymbol{u}}) \\
            & \geq \lambda Q_{tot}^{(1)}(\boldsymbol{\tau}, \boldsymbol{u}) + (1-\lambda) Q_{tot}^{(2)}(\boldsymbol{\tau}, \boldsymbol{u}) \\
            & \geq \lambda (Q_{tot}^{(1)}(\boldsymbol{\tau}, \boldsymbol{u})+w_r(\boldsymbol{\tau}, \boldsymbol{u}) Q_r(\boldsymbol{\tau}, \boldsymbol{u})) \nonumber \label{w_r_i}\\
            & + (1-\lambda) Q_{tot}^{(2)}(\boldsymbol{\tau}, \boldsymbol{u}) \\
            & = Q_{jt}(\boldsymbol{\tau}, \boldsymbol{u})
        \end{align}
    \end{subequations}

    \noindent \eqref{w_r_i} comes from $w_r(\boldsymbol{\tau}, \boldsymbol{u})=1, \forall \boldsymbol{u} \neq \bar{\boldsymbol{u}}$ and $Q_r(\boldsymbol{\tau}, \boldsymbol{u}) \leq 0$. This inequality \eqref{igm_hpf} means that $\boldsymbol{\bar{\boldsymbol{u}}}=\left[\bar{u}_i\right]_{i=1}^n$ maximizes $Q_{jt}(\boldsymbol{\tau}, \boldsymbol{u})$ with both $Q_{tot}^{(2)}(\boldsymbol{\tau}, \boldsymbol{u})$ and $Q_{tot}^{(2)}(\boldsymbol{\tau}, \boldsymbol{u})$ satisfy the monotonicity condition. Thus $\left[Q_i\left(\tau_i, u_i\right)\right]_{i=1}^n$ satisfies the IGM for $Q_{jt}(\boldsymbol{\tau}, \boldsymbol{u})$.
\end{proof}

\section{Detailed Implementation of HPF}\label{app:B}
In this section, we show the explicit implementation of HPF, which is also described in Algorithm \ref{algo}. Here we take Add-HPF-WQ as an example to illustrate the HPF algorithm process, with other HPF variants following a similar training procedure. 

In Add-HPF-WQ, WQMIX\cite{rashid2020weighted} is chosen as $\boldsymbol{\pi}_{\alpha}$ in Eq.(4) of the main text. While the $\boldsymbol{\pi}_{\beta}$ with limited representation is QMIX\cite{rashid_qmix_2018}. Since QMIX assumes the task can be directly decomposed, which also means $Q_{jt}=Q_{tot}$ is this setting. 

First, we initialize the network parameters of the two types of value decomposition policy $\theta_{\alpha}$, $\theta_{\beta}$ as well as their target network parameters $\theta_{\alpha}^{-}$, $\theta_{\beta}^{-}$, experience replay buffer $\mathcal{D}$, learning rate $\zeta$, the size of batch sampling \textit{batch\_size} and target network \textit{update\_interval} $m$. We evaluate the utility functions of both candidate VD policies, $[Q_{i}(\tau^{i},u^{i};\theta_{\alpha})]_{i=1}^{n}$ and $[Q_{i}(\tau^{i},u^{i};\theta_{\beta})]_{i=1}^{n}$, to construct the sampling policy for composite $\Pi$ via Eq.(5)-Eq.(7) in the main text. Then the interactive trajectory will be stored in the replay buffer.

During the training, HPF extracts a small batch of interaction trajectories from $\mathcal{D}$ and estimates the utility function $[Q_{i}^{\boldsymbol{\pi}_{\alpha}}]_{i=0}^{n}$, $[Q_{i}^{\boldsymbol{\pi}_{\beta}}]_{i=0}^{n}$ and the centralized value function $Q_{jt}(\theta_{\alpha})$, $Q_{jt}(\theta_{\beta})$ for each time step. The $Mlp(\cdot)$ function of Algorithm \ref{algo} represents a regular feedforward neural network that does not manipulate the weight parameters generated by the hypernetwork in any way. While $Mix(\cdot)$ function represents the mixing network in QMIX, which involves taking the absolute values of the weight generated by the hypernetwork, ensuring all weights remain positive to meet IGM principles. Then we can update $Q_{jt}(\theta_{\alpha})$, $Q_{jt}(\theta_{\beta})$ and all the utility functions like DQN\cite{mnih_human-level_2015} updates through the back-propagation. In the following part, we omit the parameters in the functions, such as $Q_{jt}^{\boldsymbol{\pi}_{\alpha}}=Q_{jt}^{\boldsymbol{\pi}_{\alpha}}(\theta_{\alpha})$, $Q_{jt}^{\boldsymbol{\pi}_{\alpha}(-)}=Q_{jt}^{\boldsymbol{\pi}_{\alpha}}(\theta^{-}_{\alpha})$, and $Q_{a}^{\boldsymbol{\pi}_{\alpha}}=Q_{a}(\tau_t^a, u_t^a), a \in \{1,\ldots,n\}$ for brevity.

\begin{algorithm}[tb]
    \caption{Add-HPF-WQ for MARL}
    \label{algo}
    \textbf{Require}: $\theta_{\alpha}$, $\theta_{\beta}$, $\theta^{-}_{\alpha}$, $\theta^{-}_{\beta}$ \\
    \textbf{Initialize} $\theta^{-}_{\alpha} \leftarrow \theta_{\alpha}$, $\theta^{-}_{\beta} \leftarrow \theta_{\beta}$, $\zeta$, $\mathcal{D}$, $b$, step=0, update interval=$m$
    \begin{algorithmic}[1] 
        \WHILE{$\text{step} < \text{step}_{max}$}
            \STATE $t=0, s_{0}= \text{initial state}$;
            \WHILE{$s_{t} \ne terminal$ and $t < \text{episode limit}$}
                \FOR{each agent $a$}
                    \STATE $\tau_{t}^{a} = \tau_{t-1}^{a} \cup \{(o_{t}, u_{t-1})\}$;
                    \STATE Select action $u^{a}_{t}$ by $\epsilon$-greedy scheme;
                    \STATE Get the utilities $Q^{\boldsymbol{\pi}_\alpha}_{a}, Q^{\boldsymbol{\pi}_\beta}_{a}, a \in \{1, \ldots, n\}$;
                    \STATE Obtain $\boldsymbol{Q}^{\boldsymbol{\pi}_{\alpha}}$ and $\boldsymbol{Q}^{\boldsymbol{\pi}_{\beta}}$ via Eq.(7) in the main text
                    \STATE Calculate $P_{\mathbf{w}}$ and sampling $\mathbf{u}_{t}$ via Eq.(5)-Eq.(6) in the main text
                \ENDFOR
                \STATE Get reward $r_{t}$ and next state $s_{t+1}$;
                \STATE $\mathcal{D}=\mathcal{D} \cup\left\{\left(s_{t}, \mathbf{u}_{t}, r_{t}, s_{t+1}\right)\right\}; t=t+1$;
                \STATE $\text{step}=\text{step}+1$;
            \ENDWHILE
            \IF{$|\mathcal{D}| > \text{batch-size}$}
                \STATE $b \leftarrow \text{select random batch of episodes from} \  \mathcal{D}$;
                \FOR{$t$ in each trajectory in batch $b$}
                    \STATE $Q_{jt}^{\boldsymbol{\pi}_{\alpha}}= \text{Mlp}\left(Q_{1}^{\boldsymbol{\pi}_{\alpha}}, \ldots, Q_{n}^{\boldsymbol{\pi}_{\alpha}},  s_{t}\right)$;
                    \STATE $Q_{tot}^{\boldsymbol{\pi}_{\beta}}= \text{Mix}\left(Q_{1}^{\boldsymbol{\pi}_{\beta}}, \ldots, Q_{n}^{\boldsymbol{\pi}_{\beta}},  s_{t}\right)$;
                    \STATE Estimate targets $Q_{jt}^{\boldsymbol{\pi}_{\alpha}(-)}$ and $Q_{tot}^{\boldsymbol{\pi}_{\beta}(-)}$;
                \ENDFOR
                \STATE Update the joint action value function:
                \STATE $\Delta Q_{jt}^{\boldsymbol{\pi}_{\alpha}}=Q_{jt}^{\boldsymbol{\pi}_{\alpha}(-)}-Q_{jt}^{\boldsymbol{\pi}_{\alpha}}$,
                \STATE $\Delta \theta_{\alpha}=\nabla_{\theta_{\alpha}}\left(\Delta Q_{jt}^{\boldsymbol{\pi}_{\alpha}}\right)^{2}$, $\theta_{\alpha}=\theta_{\alpha} - \zeta \Delta \theta_{\alpha}$;
                \STATE $\Delta Q_{tot}=Q_{tot}^{\boldsymbol{\pi}_{\beta}(-)}-Q_{tot}^{\boldsymbol{\pi}_{\alpha}}$,
                \STATE $\Delta \theta_{\beta}=\nabla_{\theta_{\beta}}\left(\Delta Q_{tot}^{\boldsymbol{\pi}_{\beta}}\right)^{2}$, $\theta_{\beta}=\theta_{\beta} - \zeta \Delta \theta_{\beta}$;
            \ENDIF
            \IF{$\text{step} \ \% \ m = 0$}
                \STATE Update the target network $\theta^{-}_{\alpha} \leftarrow \theta_{\alpha}, \theta^{-}_{\beta} \leftarrow \theta_{\beta}$;
            \ENDIF
        \ENDWHILE
    \end{algorithmic}
\end{algorithm}

\section{Experiment and Hyperparameters Setting}\label{app:C}
In this section, we thoroughly describe the experimental setups and parameter configurations for the algorithms employed in this paper. The following subsections elaborate on the specific objectives of each experimental scenario and the neural network parameters of the compared algorithms.
\subsection{Description of the experimental scenarios}
\noindent {\bf Matrix Game.} The matrix game is a simple mathematical model used to study the interaction and decision-making scenarios of two agents in a shared environment. In this setting, each agent has three possible actions, and their choices result in simultaneous outcomes according to a shared payoff function. This symmetric matrix game model exhibits an optimal joint action $(u_1, u_1)$, where agents possess an observation with all values set to 1, indicating that agents solely consider interactive actions to impact specific cooperative rewards. Agents can only choose one action per round of interaction. This model allows for a rapid evaluation of whether the tested algorithms satisfy the IGM condition and whether the centralized value function under the joint actions of agents can accurately estimate the value functions.

\noindent {\bf Predator-Prey.} The \emph{relative-overgeneralization} is formulated by a predator-prey task that is modeled in a partially observable grid-world: 8 agents try to coordinate to hunt 8 prey in a $10 \times 10$ grid. Each agent possesses a $5 \times 5$ sub-grid sight range. The prey will be caught if two adjacent predators are executing the \emph{catch} action. The predators are rewarded $r=10$ when they capture prey successfully but get punished reward $p=-2$ when only a single agent attempts to capture any prey. Collaboration among the predators is essential to capture all prey successfully and prevent incurring penalties.

\noindent {\bf StarCraft Multi-Agent Challenge.} StarCraft \uppercase\expandafter{\romannumeral2}, being a real-time strategy game, presents an excellent opportunity to address cooperative challenges in the multi-agent domain. SMAC\cite{samvelyan_starcraft_2019} utilizes Blizzard's StarCraft \uppercase\expandafter{\romannumeral2} Machine Learning API and DeepMind's PySC2 as an interface to engage with the game environment. Our focus lies on the micromanagement challenges within SMAC, where each unit is under the control of an independent agent, relying solely on limited local observation. These unit groups must be trained to engage an opposing army under the centralized control of the game's built-in scripted AI. The objective of cooperative allies is to maximize damage dealt to enemy units while minimizing received damage. Symmetrically placed opposing groups with varied initial unit placements characterize the scenario across episodes.

\subsection{Hyperparameters Setting}
For the sake of consistency, we integrate all algorithms in PyMARL and apply the same network architecture and hyperparameters to the contrasting methods. In alignment with PyMARL, the agent architecture comprises a DRQN with a recurrent layer, featuring a 64-dimensional hidden state. Throughout the entire training process, each agent employs an independent $\epsilon$-greedy exploration action selection based on its own $Q_{i}$. The $\epsilon$ value is annealed from 1.0 to 0.05 over 50$k$ time steps and remains constant thereafter (100$k$ for super-hard scenarios 6h\_vs\_8z and 3s5z\_vs\_3s6z). The discount factor, $\gamma = 0.99$, remains constant across all experiments. RMSprop with $\alpha=0.99$ is used without incorporating weight decay or momentum. The learning rate for training is consistently set to $5\times10^{-4}$ in all testing scenarios in Predator-Prey and baselines in SMAC. We apply layer-normalization to the observations, and to ensure stability when updating the composite policy $\Pi=\left[\boldsymbol{\pi}_\alpha, \boldsymbol{\pi}_\beta \right]$ with sampled trajectories, the Adam optimizer was utilized in the HPF scheme with the default setting, with the temperature $\eta=1$ in candidate policies sampling in HPF.
\footnote{We use the SC2.4.10 version of SMAC through all the testing scenarios.}

As for The Matrix Game, all agents adhere to full exploration, i.e., $\epsilon=1$, and remain unchanged during training. The WQMIX baseline and the WQMIX policy in HPF employ a sub-optimal joint action weight of $w=0.1$, while in the SMAC environment, it opts for $w=0.75$ from the range $w=(0.5, 0.75)$. While the QPLEX baseline and the QPLEX policy in HPF enjoy the same setting as in its open-source code for fair comparison. 

Training halts every 10000-time steps for 16 evaluation episodes. The most recent 5000 training episodes are maintained in the replay buffer, and batches of 32 episodes are uniformly sampled for training. All agent networks share the parameters of networks, and we update the target networks every 200 episodes. To ensure fairness in comparing the performance, all other parameters of the baselines are consistent with their corresponding settings in the official code bases.

\begin{figure}[!b]
    \centering
    \includegraphics[width=0.7\linewidth]{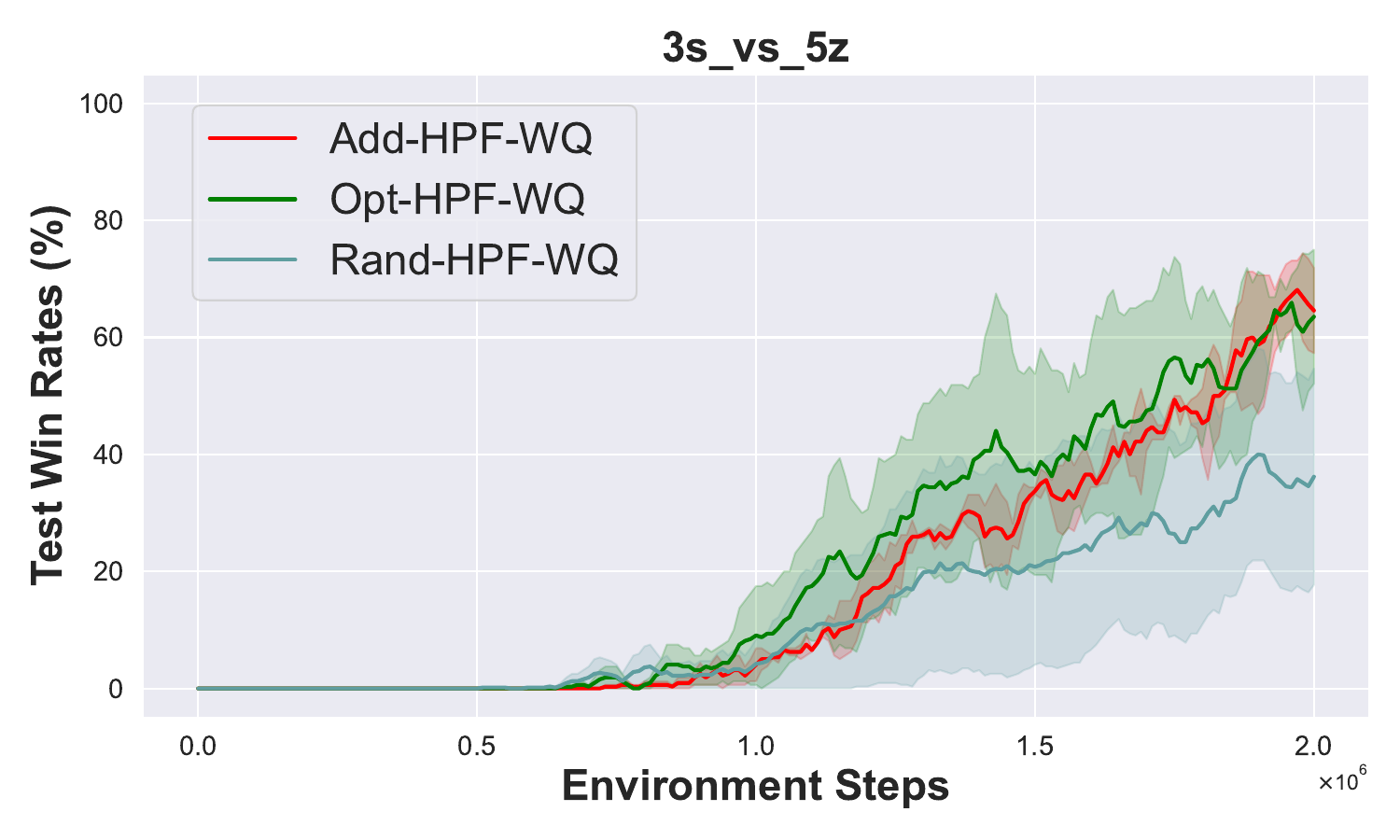}
    \includegraphics[width=0.7\linewidth]{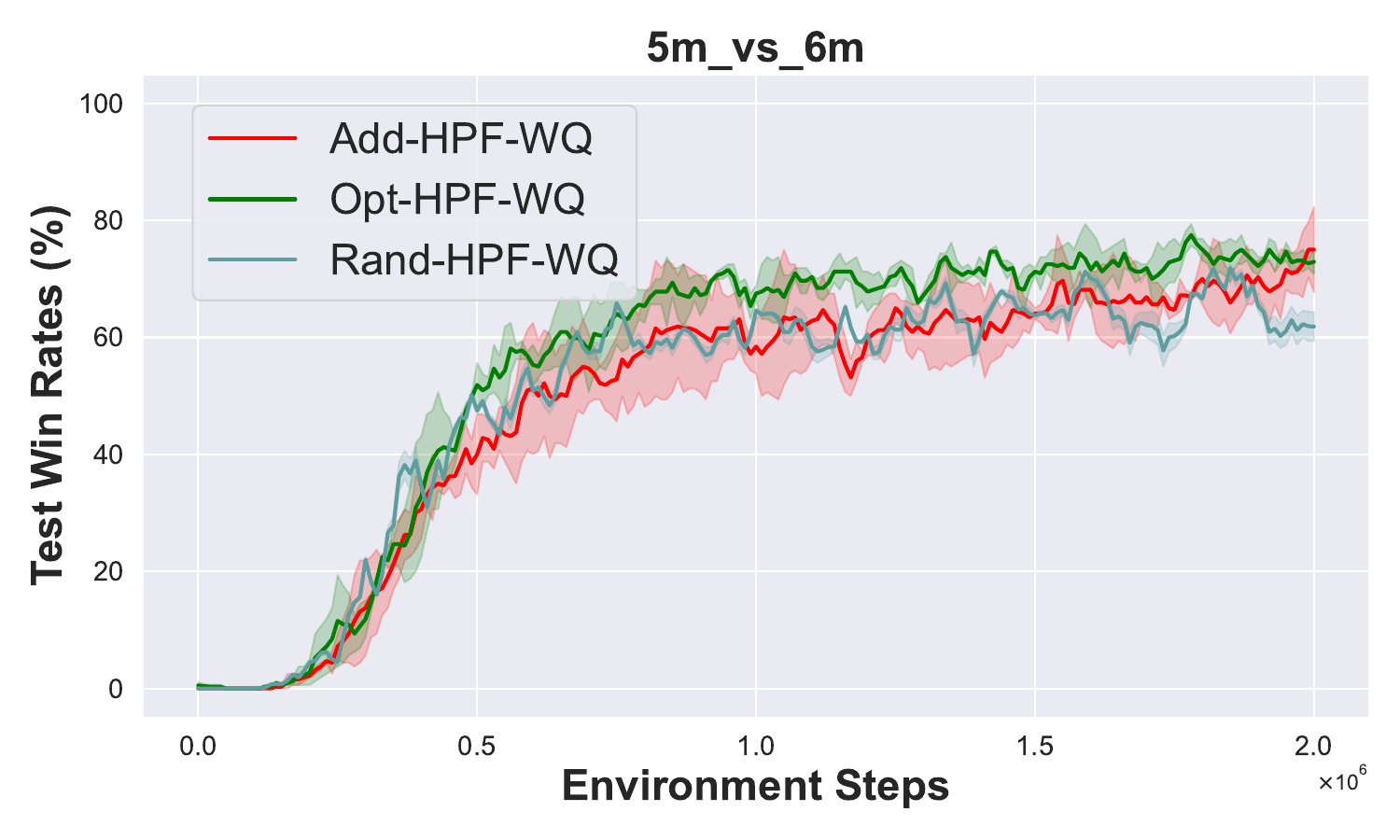}
    \includegraphics[width=0.7\linewidth]{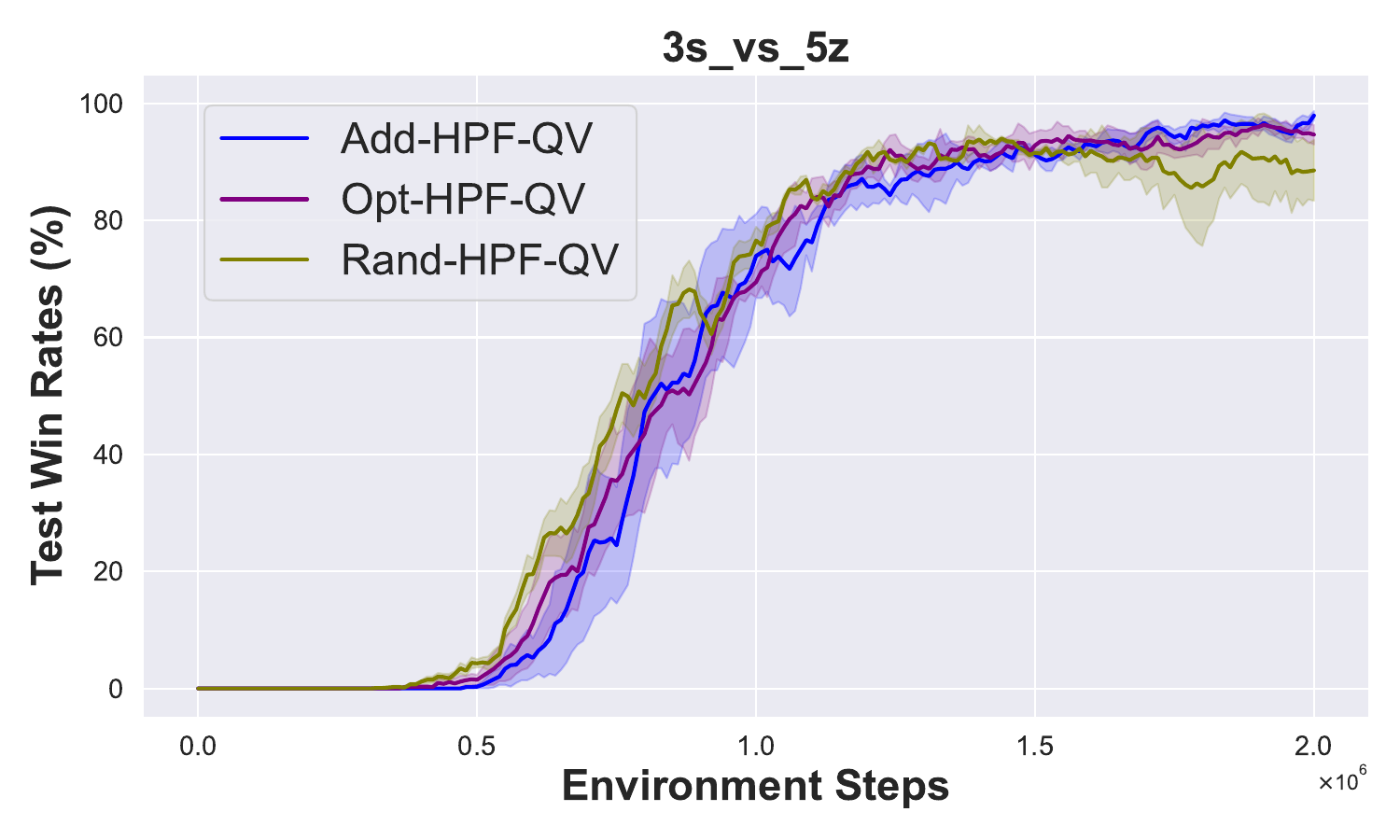}
    \includegraphics[width=0.7\linewidth]{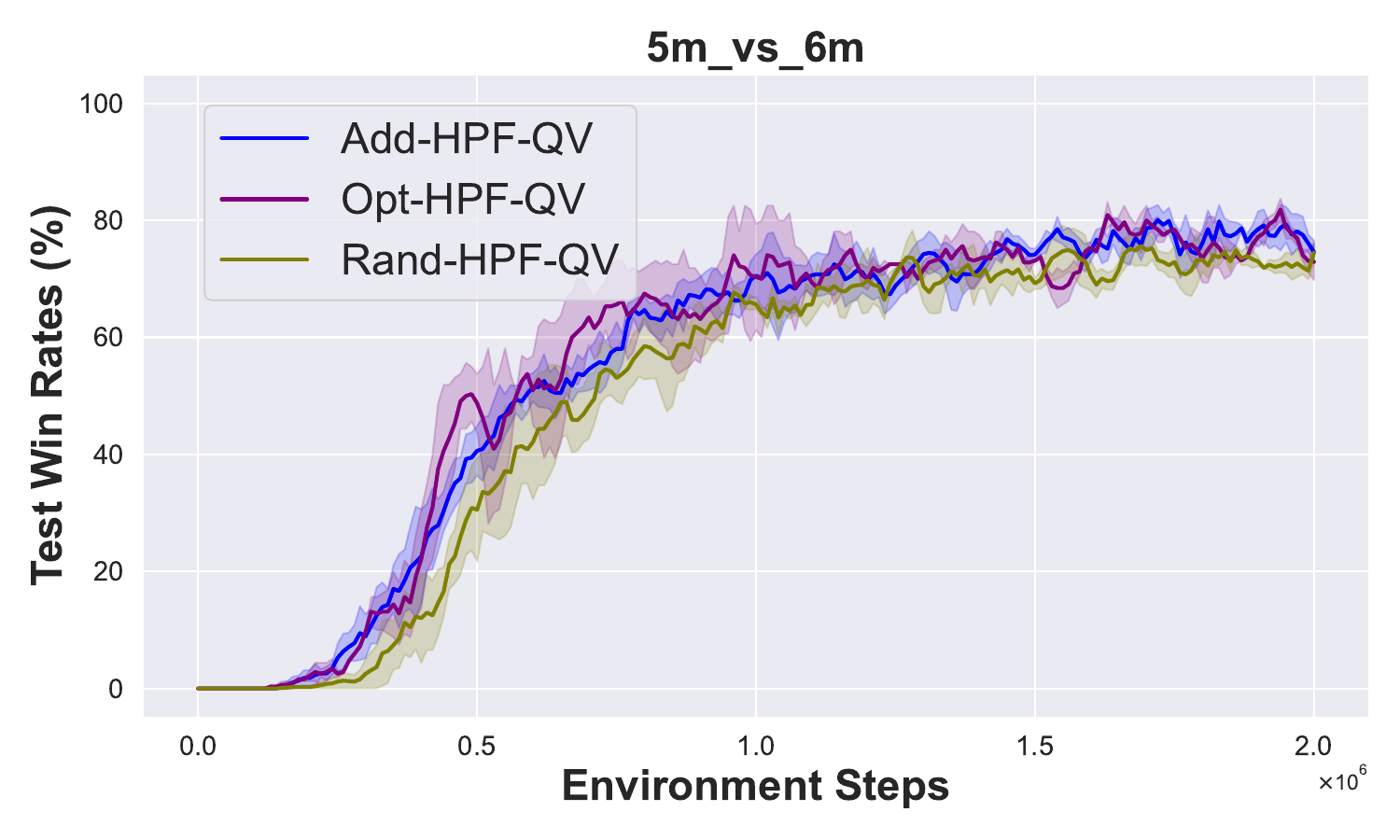}
    \caption{Ablation studies of the random candidate VD policy sampling in HPF.}
    \label{ab_results}
\end{figure}

\begin{figure}[!b]
    \centering
    \includegraphics[width=0.8\linewidth]{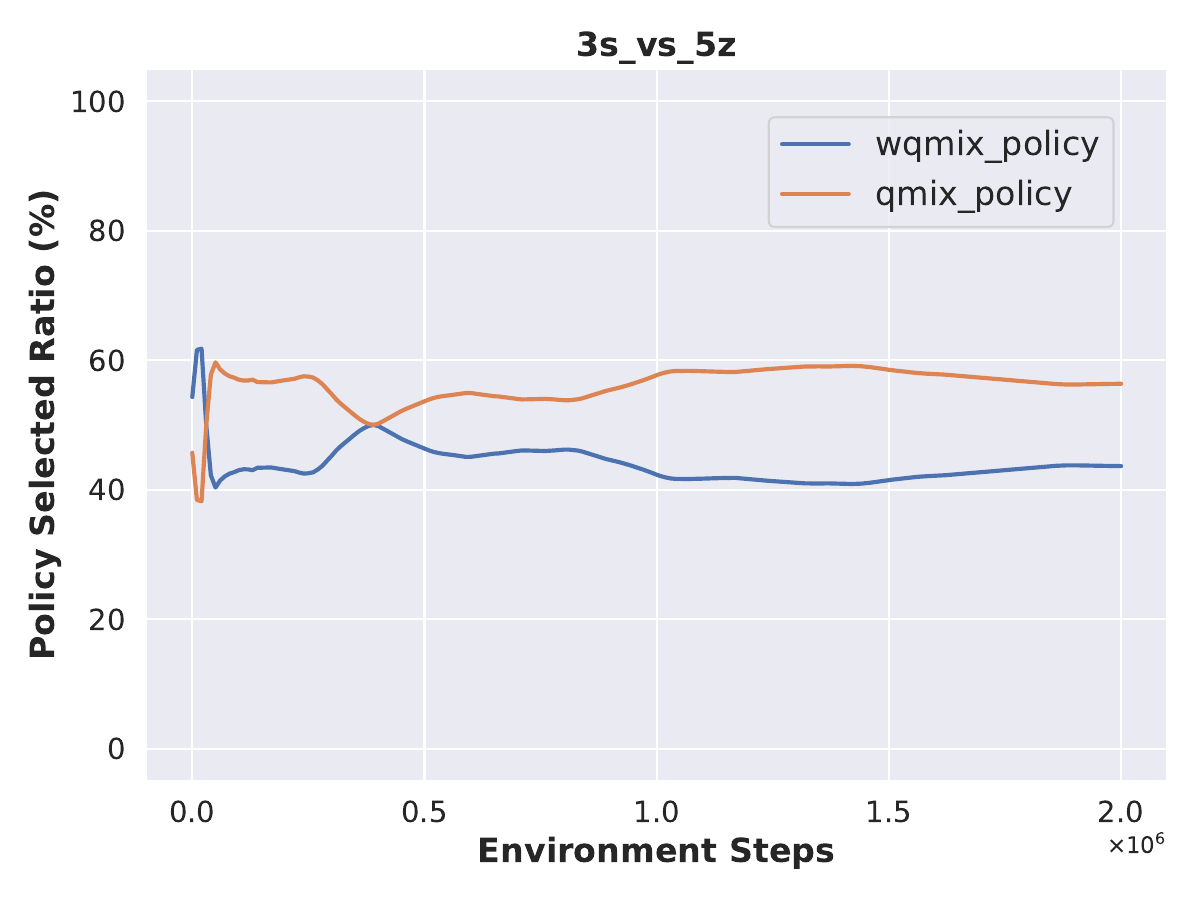}
    \includegraphics[width=0.8\linewidth]{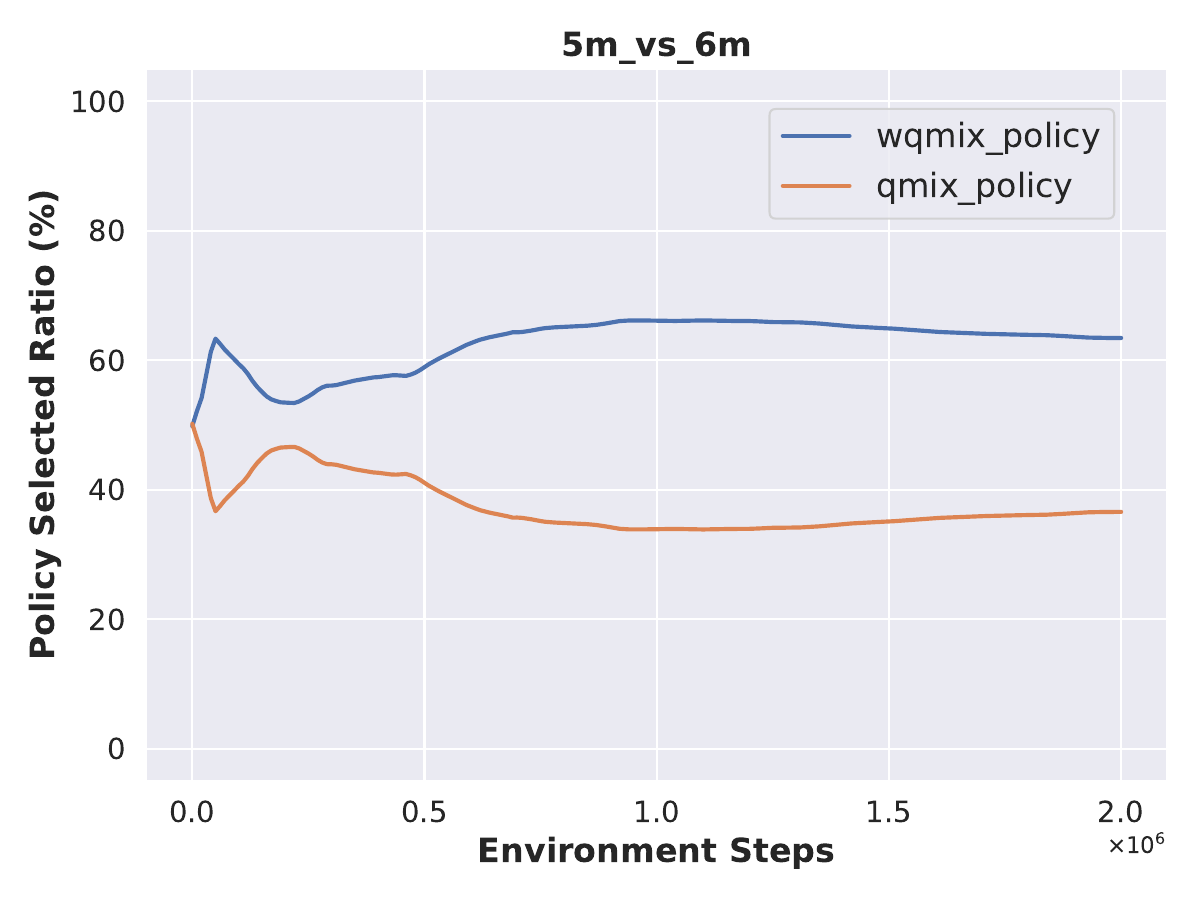}
    \includegraphics[width=0.8\linewidth]{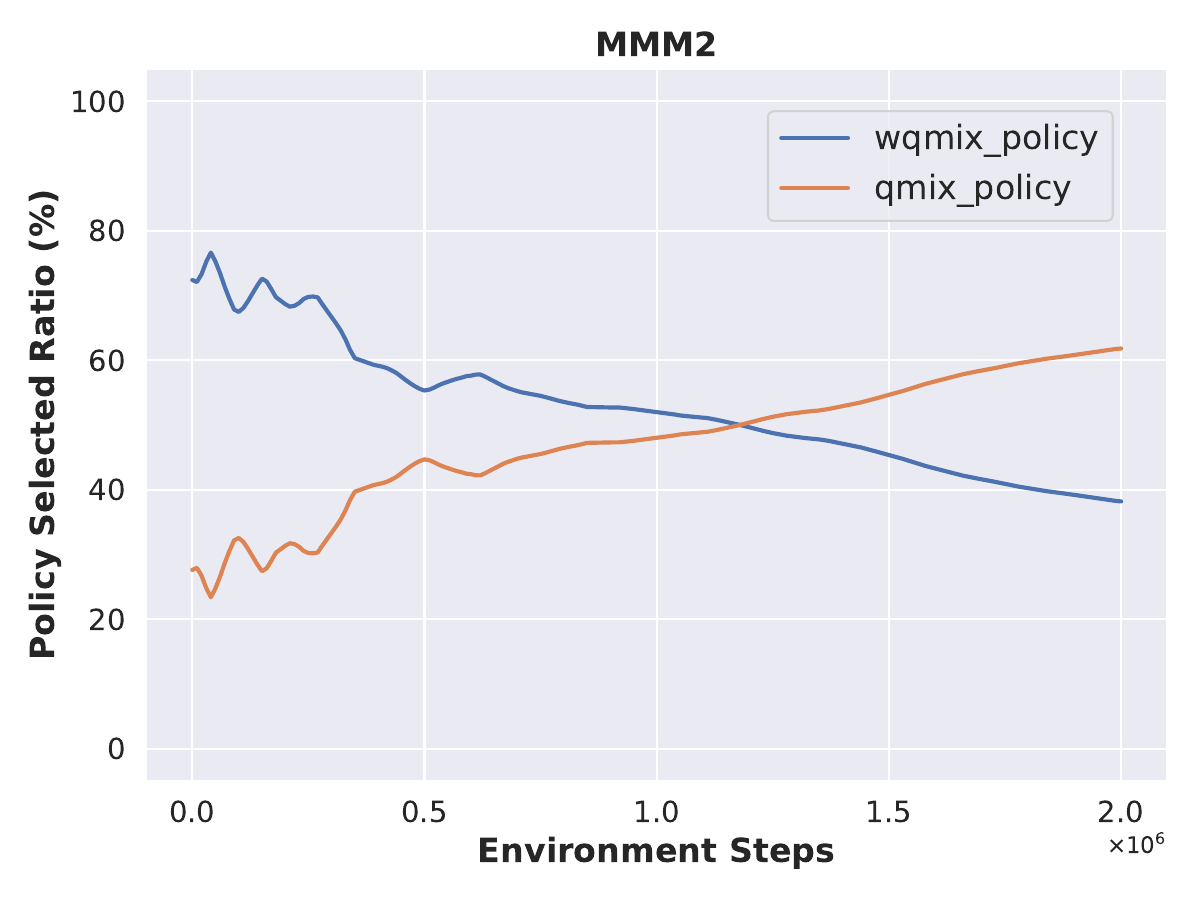}
    \caption{The selection probabilities of different VD policies in HPF along with training process.}
    \label{ab_selection}
\end{figure}

\section{Additional Experimental Results}\label{app:D}
In this section, we supplement the ablation experiments omitted in the main text. Specifically, we replace the sampling method of the composite policy in the HPF scheme with random sampling to verify the importance of utilizing value estimates in VD methods as sampling weights in HPF. Additionally, we record the sampling frequency of heterogeneous VD methods within HPF across different scenarios, attempting to explain which type of VD method played a greater role in model updates.

\subsection{Omitted Ablation Studies}
This subsection provides the ablation experimental results of a random sampling of the heterogeneous VD method in the HPF scheme. As illustrated in Figure \ref{ab_results}, when HPF adopts random sampling to create a composite policy, its performance in the 3s\_vs\_5z and 5m\_vs\_6m scenarios is inferior to the approach proposed in this paper, which utilizes the sampling from a Boltzmann policy based on value function estimates of the heterogeneous VD method, as depicted by Eq.(5)-Eq.(6) in the main text. Since a greater estimate implies that the currently sampled VD method has a higher cumulative expected return, which makes the data obtained by the interaction of this method with the environment have a higher potential value for the model update. This indicates that the composite policy of HPF is capable of effectively integrating the merits of the heterogeneous VD method and thereby acquiring better interaction experiences, consequently enhancing the training efficiency of the model.

\subsection{Additional Analysis}
The main purpose of this subsection is to analyze which type of VD policy is more likely to be chosen in HPF, providing insights into the impact of interaction data generated by different VD methods on the model. We recorded the probability of sampling a certain type of VD policy during model training in Add-HPF-WQ, as illustrated in Figure \ref{ab_selection}. Here, \textit{wqmix\_policy} represents the probability of choosing the WQMIX policy, while \textit{qmix\_policy} represents the probability of choosing the QMIX policy. The selection of VD policies by the HPF scheme varies across different cooperative tasks. However, a general observation is that in tasks where the WQMIX policy consistently learns effective collaborative behaviors, such as 5m\_vs\_6m, HPF tends to choose and maintain the WQMIX policy. While in the tasks where obtaining effective information from WQMIX interaction experiences proves challenging, HPF gradually leans towards selecting the QMIX policy for interaction, thereby enhancing the training effectiveness of the model. For instance, in the case of 3s\_vs\_5z, the probability of choosing the WQMIX policy increases at the beginning. Still, as it fails to learn suitable cooperative behaviors, HPF gradually shifts towards selecting the QMIX policy for interaction. In more complex scenarios such as MMM2, HPF progressively increases the probability of selecting the QMIX policy for interaction. This suggests that the HPF's approach of extending VD methods can effectively perceive the task complexity of the environment and adaptively adjust the sampling probability of interaction, enabling the model to learn superior collaborative policies.

\end{document}